%% file: archive.tex
\documentclass{acm_proc_article-sp}
\input{macros}

\usepackage{graphicx}
\usepackage{algorithm}
\usepackage[noend]{algpseudocode}

\begin{document}

\title{Opacity Proof for CaPR+ Algorithm}

%
\numberofauthors{3}
\author{
\alignauthor
Anshu S Anand \\
       \affaddr{Homi Bhabha National Institute, Mumbai}\\
       \email{\texttt{anshusanand2001@gmail.com}}
\alignauthor
R K Shyamasundar\\
       \affaddr{Indian Institute of Technology}\\
       \affaddr{Mumbai}\\
       \email{\texttt{shyamasundar@gmail.com}}
\alignauthor
Sathya Peri\\
       \affaddr{Indian Institute of Technology}\\
       \affaddr{Hyderabad}\\
       \email{\texttt{sathya\_p@iith.ac.in}}
}

\maketitle
\begin{abstract}
In this paper, we describe an enhanced Automatic Checkpointing and Partial Rollback algorithm($CaPR^{+}$) to realize 
Software Transactional Memory(STM) that is based on
continuous conflict detection, lazy versioning with automatic checkpointing, and partial rollback. 
Further, we provide a proof of correctness of $CaPR^{+}$ algorithm, in particular, Opacity, a STM
correctness criterion, that precisely
captures the intuitive correctness guarantees required of transactional 
memories. The algorithm provides a natural way to realize a hybrid system of pure aborts and partial rollbacks.
We have also implemented the algorithm,
and shown its effectiveness with reference to the Red-black tree micro-benchmark and STAMP benchmarks.  
The results obtained demonstrate the effectiveness of the Partial Rollback mechanism over pure abort 
mechanisms, particularly in applications consisting of large transaction lengths.
\end{abstract}




\section{Introduction}
The challenges posed by the use of low-level synchronization primitives like 
locks led to the search of alternative parallel programming
 models to make the process of writing concurrent programs easier. Transactional 
Memory is a promising programming memory in this regard.
 
A Software Transactional Memory(STM)\cite{STM} is a concurrency control mechanism that resolves data conflicts 
in software as compared to in hardware by HTMs. 

STM provides the programmers with high-level
constructs to delimit transactional operations and with these constructs in 
hand, the
programmer just has to demarcate atomic blocks of code, that identify critical
regions that should appear to execute atomically and in isolation from other 
threads.
The underlying transactional memory implementation then implicitly takes care of
the correctness of concurrent accesses to the shared data. The STM might 
internally use fine-grained locking, or some non-blocking mechanism, but this is hidden 
from the programmer and the application thereby relieving him of the burden of handling concurrency issues.

Several STM implementations have been proposed, which are mainly classified
based on the following metrics:\newline
1) shared object update(version management) - decides when does a transaction 
update its shared objects during its lifetime. \newline
2) conflict detection - decides when does a transaction detect a conflict with 
other transactions in the system. \newline
3) concurrency control - determines the order in which the events - conflict, 
its detection and resolution occur in the system.

Each software transaction can perform operations on shared data, and then either
commit or abort. When the transaction commits, the effects of all its operations
become immediately visible to other transactions; when it aborts, all its 
operations are rolled back and none of its effects are visible to other transactions. Thus, abort is an important STM mechanism that 
allows the transactions to be atomic.  However, abort comes at a cost, as an abort operation implies additional overhead as the transaction 
is required to be re-executed after canceling the effects of the local transactional operations. Several solutions have been proposed
for this, that are based on partial rollback, where the transaction rolls back to an intermediate consistent state rather than restarting from beginning. 
\cite{checkpt} was the first work that illustrated the use of checkpoints in boosted transactions and \cite{wali} suggested using checkpoints in HTMs.
In \cite{lupei} the partial rollback operation is based only on shared data that does not support local data which requires extra effort from 
the programmer in ensuring consistency. \cite{monika1} and \cite{monika2} is an STM algorithm that supports both shared and local data for 
partial rollback. \cite{alice} is another STM that supports both shared and local data. 
Our work is based on \cite{monika1}. We present an improved and simplified algorithm, Automatic Checkpointing and Partial Rollback algorithm($CaPR^{+}$)
and prove its correctness.

Several correctness criteria exist for STMs like linearizability, 
serializability, rigorous scheduling,
 etc. However, none of these criteria is sufficient to describe the semantics of 
TM with its subtleties. Opacity is a criterion that
captures precisely the correctness requirements that have been intuitively 
described by many TM designers. We discuss Opacity in section~\ref{sec:model} and present the proof of opacity of $CaPR^{+}$ algorithm
in section~\ref{sec:proof}.

\section{System Model}
\label{sec:model}

The notations defined in this section have been inspired from 
\cite{KuzPer:CORR:NI:2012}. We assume a system
of $n$ processes (or threads), $p_1,\ldots,p_n$ that access a collection of 
\emph{objects} via atomic \emph{transactions}.
The processes are provided with the following \emph{transactional operations}: 
\textit{\begtrans}$()$ \op, which 
invokes a new transaction and returns the $id$ of the new transaction; the 
\textit{write}$(x,v,i)$ operation that 
updates object $x$ with value $v$ for a transaction $i$, the \textit{read}$(x)$ 
operation that returns a value read in
$x$, \textit{tryC}$()$ that tries to commit the transaction and
returns \textit{commit} ($c$ for short) or \textit{abort} ($a$ for
short), and \textit{\trya}$()$ that aborts the transaction and returns
$A$. The objects accessed by the read and write \op{s} are called as
\tobj{s}. For the sake of presentation simplicity, we assume that the
values written by all the transactions are unique. 

Operations \textit{write}, \textit{read} and \textit{\tryc} may
return $a$, in which case we say that the operations \emph{forcefully
abort}. Otherwise, we say that the operation has \emph{successfully}
executed.  Each operation is equipped with a unique transaction
identifier. A transaction $T_i$ starts with the first operation and
completes when any of its operations returns $a$ or $c$. 
Abort and commit \op{s} are called \emph{\termop{s}}. 
For a transaction $T_k$, we denote all its  read \op{s} as $Rset(T_k)$
and write \op{s} $Wset(T_k)$. Collectively, we denote all the \op{s}
of a  transaction $T_i$ as $\evts{T_k}$. 

\vspace{1mm}
\noindent
\textit{Histories.} 
A \emph{history} is a sequence of \emph{events}, i.e., a sequence of
invocations and responses of transactional operations. The collection
of events is denoted as $\evts{H}$. For simplicity, we only consider
\emph{sequential} histories here: the invocation of each transactional
operation is immediately followed by a matching response. Therefore,
we treat each transactional operation as one atomic event, and let
$<_H$ denote the total order on the transactional operations incurred
by $H$. With this assumption the only relevant events of a transaction
$T_k$ are of the types: $r_k(x,v)$, $r_k(x,A)$, $w_k(x, v)$, $w_k(x,
v,A)$, $\tryc_k(C)$ (or $c_k$ for short), $\tryc_k(A)$, $\trya_k(A)$ (or $a_k$ 
for short). 
We identify a history
$H$ as tuple $\langle \evts{H},<_H \rangle$. 

Let $H|T$ denote the history consisting of events of $T$ in $H$, 
and $H|p_i$ denote the history consisting of events of $p_i$ in $H$. 
We only consider \emph{well-formed} histories here, i.e.,
(1) each $H|T$ consists of  a read-only prefix (consisting of read
operations only), followed by a
write-only part (consisting of write operations only), possibly \emph{completed}
with a $\tryc$ or $\trya$ operation\footnote{This restriction brings no loss of 
generality~\cite{KR:2011:OPODIS}.},  and
(2) each $H|p_i$ consists of a sequence of transactions, where no new
transaction begins before the last transaction
completes (commits or a aborts). 

We assume that every history has an initial committed transaction $T_0$
that initializes all the data-objects with 0. The set of transactions
that appear in $H$ is denoted by $\txns(H)$. The set of committed
(resp., aborted) transactions in $H$ is denoted by $\comm(H)$
(resp., $\aborted(H)$). The set of \emph{incomplete} or \emph{live} transactions
in $H$ is denoted by $\id{incomplete}(H)$
($\id{incomplete}(H)=\txns(H)-\comm(H)-\aborted(H)$). 

For a history $H$, we construct the \emph{completion} of $H$, denoted 
$\overline{H}$, 
by inserting $a_k$ immediately after the last event 
of every transaction $T_k\in\id{incomplete}(H)$.

\vspace{1mm}
\noindent
\textit{Transaction orders.} For two transactions $T_k,T_m \in \txns(H)$, we say 
that  $T_k$ \emph{precedes} $T_m$ in the \emph{real-time order} of $H$, denote 
$T_k\prec_H^{RT} T_m$, if $T_k$ is complete in $H$ and the last event of $T_k$ 
precedes the first event of $T_m$ in $H$. If neither $T_k\prec_H^{RT} T_m$ nor 
$T_m \prec_H^{RT} T_k$, then $T_k$ and $T_m$ \emph{overlap} in $H$. A history 
$H$ is \emph{t-sequential} if there are no overlapping transactions in $H$, 
i.e., every two transactions are related by the real-time order.

For two transactions $T_k$ and $T_m$ in $\txns(H)$, we say that \emph{$T_k$ 
precedes $T_m$ in conflict order}, denoted $T_k \prec_H^{CO} T_m$ if: (a) (w-w 
order) $c_k <_H c_m$ and $Wset(T_k) \cap Wset(T_m) \neq \emptyset$; (b) (w-r 
order) $c_k <_H r_m(x,v)$, $x \in Wset(T_k)$ and $v \neq A$; (c) (r-w order) 
$r_k(x,v)<_H c_m$ and $x\in Wset(T_m)$ and $v \neq A$. Thus, it can be seen that 
the conflict order is defined only on \op{s} that have successfully executed.

\cmnt {
\vspace{1mm}
\noindent
\textit{Sub-histories.} A \textit{sub-history}, $SH$ of a history
$H$ denoted as the tuple $\langle \evts{SH},$ $<_{SH}\rangle$ and is
defined as: (1) $<_{SH} \subseteq <_{H}$; (2) $\evts{SH} \subseteq
\evts{H}$; (3) If an event of a transaction $T_k \in \txns(H)$ is in $SH$ then 
all
the events of $T_k$ in $H$ should also be in $SH$. 
For a history
$H$, let $R$ be a subset of $txns(H)$, the transactions in
$H$. 
Then $\shist{R}{H}$ denotes  the \ssch{} of $H$ that is
formed  from the \op{s} in $R$. 
}

\vspace{1mm}
\noindent
\textit{Valid and legal histories.} 
Let $H$ be a history and $r_k(x, v)$ be a read {\op} in $H$. A
successful read $r_k(x, v)$ (i.e $v \neq A$), is said to be
\emph{\valid} if there is a transaction $T_j$ in $H$ that commits
before $r_K$ and $w_j(x, v)$ is in $\evts{T_j}$. Formally, $\langle
r_k(x, v)$  is \valid{} 
$\Rightarrow \exists T_j: (c_j <_{H} r_k(x, v)) \land (w_j(x, v) \in
\evts{T_j}) \land (v \neq A) \rangle$.  The history $H$ is \valid{} 
if all its successful read \op{s} are \valid. 

We define $r_k(x, v)$'s \textit{\lastw{}} as the latest commit event
$c_i$ such that $c_i$ precedes $r_k(x, v)$ in $H$ and $x\in\Wset(T_i)$
($T_i$ can also be $T_0$). A successful read \op{} $r_k(x, v)$ (i.e
$v \neq A$), is said to be \emph{\legal{}} if transaction $T_i$
(which contains  $r_k$'s \lastw{}) also writes $v$ onto $x$. Formally,
$\langle r_k(x, v)$ \textit{is \legal{}} $\Rightarrow (v \neq A) \land
(\lwrite{r_k(x, v)}{H} = c_i) \land (w_i(x,v) \in \evts{T_i})
\rangle$.  The history $H$ is \legal{} if all its successful read
\op{s} are \legal. Thus from the definitions we get that if $H$ is \legal{} then 
it is also \valid.


\vspace{1mm}
\noindent
\textit{Opacity.} 
We say that two histories $H$ and $H'$ are \emph{equivalent} if
they have the same set of events.  
Now a history $H$ is said to be \textit{opaque}
\cite{GuerKap:2008:PPoPP,tm-book} 
if $H$ is \valid{}  and there exists a
t-sequential legal history $S$ such that (1) $S$ is equivalent to
$\overline{H}$ and (2) $S$ respects $\prec_{H}^{RT}$, i.e
$\prec_{H}^{RT} \subset \prec_{S}^{RT}$. 
By requiring $S$ being equivalent to $\overline{H}$, opacity treats
all the incomplete transactions as aborted. 

\cmnt {
Along the same lines, a \valid{} history $H$ is said to be
\textit{strictly serializable} if $\shist{\comm(H)}{H}$ is opaque.
Thus, unlike opacity, strict serializability does not include aborted
transactions in the global serialization order.
}

\vspace{1mm}
\noindent
\textit{Implementations and Linearizations.} A (STM) implementation is
typically a library of functions for implementing: $read_k$,
$write_k$, $\tryc_k$ and $\trya_k$ for a transaction $T_k$.  
\cmnt{Consider an implementation $M$. We denote the set of all histories 
generated by $M$ as $\gen{M}$. }We say that an implementation $M_p$ is correct 
w.r.t to a property $P$ if all the histories generated by $M_p$ are in 
$P$\cmnt{, i.e. $\gen{M_p} \in P$}. 
The histories generated by an STM implementations are normally not
sequential, i.e., they may have overlapping transactional \op{s}. 
Since our correctness definitions are proposed for sequential histories,
to reason about correctness of an implementation, we order the events in
a non-concurrent history in a sequential manner. The ordering must
respect the real-time ordering of the \op{s} in the original history. 
In other words, if the response \op{} $o_i$ occurs before the
invocation \op{} $o_j$ in the original history then $o_i$ occurs
before $o_j$ in the sequential history as well. 
Overlapping events, i.e. events whose invocation and response events
do not occur either before or after each other, can be ordered in any way. 

We call such an ordering as \emph{linearization} \cite{HerlWing:1990:TPLS}. Now 
for a (non-sequential) history $H$ generated by an implementation $M$, multiple 
such linearizations are possible. An implementation $M$ is considered 
\emph{correct} (for a given correctness property $P$)  if every its history has 
a correct linearization (we say that this linearization is exported by $M$).  

We assume that the implementation has enough information to generate an unique 
linearization for H to reason about its correctness. For instance, 
implementations that use locks for executing conflicting transactional 
operations, the order of access to locks by these (overlapping) operations can 
decide the order in obtaining the sequential history. This is true with STM 
systems such as \cite{ImbsRay:2008:OPODIS, CIR11, AttHill:TCS:SingMult:2012} 
which use locks.

 \section{CaPR+ Algorithm} \label{capra}

 In this section, we present the data structures and the $CaPR^{+}$ Algorithm. 
The various data structures used in the $CaPR^{+}$ Algorithm are categorized into local workspace and global workspace,
depending on whether the data structure is visible to the local transaction or every transaction. 
The data structures used in the local workspace are as follows: 

\begin{enumerate}
\item Local Data Block(LDB) - Each entry consists of the 
local object and its current value in the transaction(Table 1).
\item Shared object Store(SOS) - An entry in Table 2 stores the address of the shared object, its value, a read flag and a write flag.
Both read and write flags have 0 as the initial value. Value 1 in read/write flag indicates the object has been read/written 
by the transaction, respectively. 
%
\item Checkpoint Log(Cplog) - Used to partially rollback a transaction as shown in Table 3, where each entry stores, 
a) the shared object whose read initiated the log entry (this entry is made every 
time a shared object is read for the first
time by the transaction), b) program location from where the transaction 
should proceed after a rollback, and c) the current snapshot of the transaction's local data block and the shared object 
store.
\end{enumerate}

\begin{table}[h]
\centering
\caption{Local Data Block}
\begin{tabular}{|l|c|}
\hline
Object & Value \\ \hline
       &       \\ \hline
\end{tabular}
\end{table}

\begin{table}[h]
\centering
\caption{Shared object Store}
\begin{tabular}{|l|l|l|l|}
\hline
Object &  Current Value & Read flag & Write flag\\ \hline
         &         &   &    \\ \hline
\end{tabular}
\end{table}

\begin{table}[h]
\centering
\caption{Checkpoint Log}
\begin{tabular}{|l|c|l|}
\hline
\begin{tabular}[c]{@{}l@{}}Victim Shared\\  object\end{tabular} & Program Location      & Local Snapshot \\ \hline
               &                  &              \\ \hline
\end{tabular}
\end{table}

 \begin{table}[h]
 \centering
 \caption{Global List of Active Transactions}
\begin{tabular}{|l|c|l|}
\hline
\begin{tabular}[c]{@{}l@{}}Transaction\\ ID\end{tabular} & 
\begin{tabular}[c]{@{}c@{}}Status \\ Flag\end{tabular} & Conflict Objects \\ \hline
                &                        &              \\ \hline
\end{tabular}
\end{table}

\begin{table}[h]
\centering
\caption{Shared Memory}
\begin{tabular}{|l|c|l|}
\hline
\begin{tabular}[c]{@{}c@{}}Shared\\ object\end{tabular} & 
\begin{tabular}[c]{@{}l@{}}Value\end{tabular} & \begin{tabular}[c]{@{}l@{}}List 
of\\ active readers\end{tabular}\\ \hline
                 &            &            \\ \hline
\end{tabular}
\end{table}

The data structures in the global workspace are: 
\begin{enumerate}
\item Global List of Active Transactions(Actrans) - Each entry in this list contains a) a 
unique transaction identifier, b) a
status flag that indicates the status of the transaction, as to whether the 
transaction is in conflict with any of the
committed transactions, and c) a list of all the objects in conflict with the 
transaction. This list is updated by the
committed transactions.
\item Shared Memory(SM) - Each entry in the shared memory stores a) a shared object, 
b) its value, and c) an active readers
list that stores the transaction IDs of all the transactions reading the shared 
object.
\end{enumerate}

The $CaPR^{+}$ algorithm is shown in Algorithm 1.

\begin{algorithm}
\caption{CaPR Algorithm}
\begin{algorithmic}[1]
\Procedure{ReadTx}{$t, o, pc$}
\If{ o is in t's local data block}
\State $str.val \leftarrow \textit{o.val from LDB}$
\State $\textit{return l}  \leftarrow 1(Success);$
\ElsIf{o is in t's shared object store}
\State $str.val \leftarrow \textit{o.val from SOS}$
\State $\textit{return l} \leftarrow 1(Success);$
\ElsIf{o is in shared memory} 
\State $\textit{obtain locks on object o, \& the entry for transaction t;}$
\If{t.status\_flag = RED}
\State $\textit{Unlock lock on t and o}$
\State $PL = partially\_Rollback(t);$
\State $\textit{update str.PL = PL}$
\State $\textit{return l} \leftarrow 0(Rollback);$
\EndIf
\State $\textit{create checkpoint entry in checkpoint log for o;}$
\State $str.val \leftarrow \textit{o.val from Shared Memory}$
\State $\textit{add t to o's readers' list}$
\State $\textit{add o into SOS and set its read flag to 1;}$
\State $\textit{release locks on o and t;}$
\State $\textit{return l} \leftarrow 1(Success);$
\Else \Comment{o not in shared memory}
\State$\textit{return l} \leftarrow 2(Error);$
 \EndIf
\EndProcedure

\Procedure{WriteTx}{$o, t$}
\If{o is a local object}
\State $\textit{update o in local data block;}$
\ElsIf{o is a shared object}
\If{o is in shared object store}
\State $\textit{update o in SOS and set its write flag to 1;}$
\Else
\State $\textit{insert o in SOS and set its write flag to 1;}$
 \EndIf
  \EndIf
\EndProcedure

\Procedure{commitTx}{$t$}
\State $\textit{Assign t's write-set, }t.WS = \{o| \textit{o is in SOS and o's write flag} = 1\}$
\State $\textit{Sort all objects in t.WS, and obtain locks on them;}$
\State $\textit{Initialize A = \{t\}}$;
\State $\textit{\textbf{for each} object o in the t.WS}$
\State $\textit{\hspace{0.5cm}A = A} \cup  {\textit{active readers of o;}}$
\State $\textit{Sort all transactions in 'A', and obtain locks on them;}$
\If {t.status\_flag = RED}
\State $PL = partially\_Rollback(t);$
\State $\textit{release all locks;}$
\State $\textit{return PL;}$
\EndIf
\State $\textit{\textbf{for each} object, wo in t's write set, t.WS }$
\State $\textit{\hspace{0.5cm}update wo.value in SM from the local copy of wo.}$
\State $\textit{\hspace{0.5cm}\textbf{for each} transaction rt in wo's active readers' list, }$
\State $\textit{\hspace{1.0cm}add the objects in t.WS to transaction rt's}$
\Statex $\textit{\hspace{1.5cm}conflict objects' list;}$
\State $\textit{\hspace{1.0cm}set transaction rt's status flag to RED;}$
\State $\textit{delete t from actrans;}$
\State $\textit{\textbf{for each} object, ro in t's readers-list}$
\State $\textit{\hspace{0.5cm} delete t from ro's active readers list;}$
\State $\textit{release all locks;}$
\State $\textit{return 0;}$
\EndProcedure

\Procedure{Partially\_Rollback}{$t$}
\State $\textit{identify safest checkpoint - earliest conflicting object;}$
\State $\textit{apply selected checkpoint;}$
\State $\textit{delete t from active reader's list of all objects rolled back}$
\State $\textit{reset status flag to GREEN;}$
\State $\textit{return PL(the new program location);}$
\EndProcedure

\end{algorithmic}
\end{algorithm}

\section{Conflict Opacity}
\label{sec:co}

In this section we describe about \textit{Conflict Opacity} (\co), a subclass of 
Opacity using conflict order (defined in \secref{model}). This subclass is 
similar to conflict serializability of databases whose membership can be tested 
in polynomial time (in fact it is more close to order conflict serializability) 
\cite[Chap 3]{WeiVoss:2002:Morg}. 

\begin{definition}
\label{def:coop1}
A history $H$ is said to be \emph{conflict opaque} or \emph{\coop} if $H$ is 
\valid{} 
and there exists a t-sequential legal history $S$ such that (1) $S$ is 
equivalent to
$\overline{H}$ and (2) $S$ respects $\prec_{H}^{RT}$ and $\prec_{H}^{CO}$. 
\end{definition}

From this definition, we can see that \coop{} is a subset of \opty.

\subsection{Graph characterization of \coopty}
\label{subsec:graph}

Given a history $H$, we construct a \textit{conflict graph}, $\cg{H} =
(V,E)$ as follows:  (1) $V=\txns(H)$, the set of transactions in $H$
(2) an edge $(T_i,T_j)$ is added to $E$ whenever 
$T_i \prec_H^{RT} T_j$ or $T_i \prec_H^{CO} T_j$, i.e., whenever 
$T_i$ precedes $T_j$ in the real-time or conflict order.
\cmnt{
\begin{itemize}
\item[2.1] Real-Time edges: If $T_i$ precedes $T_j$ in $H$
\item[2.2] \co{} edges: If one of the following conditions hold (a) w-w edges: 
$c_i<_{H} c_j$ and $Wset(T_i) \cap Wset(T_j) \neq \emptyset$, (b) w-r edges: 
$c_i<_{H} r_j(x,v)$ and $x \in Wset(T_j)$, or (c) r-w edges: $r_i(x,v)<_{H} c_j$ 
and $x \in Wset(T_j)$. 
\end{itemize}
}

Note, since $\txns(H)=\txns(\overline{H})$ and 
($\prec_H^{RT}\cup\prec_H^{CO})=(\prec_{\overline{H}}^{RT}\cup\prec_{\overline{H
}}^{CO}$),
we have $\cg{H} = \cg{\overline{H}}$. 
In the following lemmas, we show that the graph characterization
indeed helps us verify the membership in \coopty.

\begin{lemma}
\label{lem:co-eq}
Consider two histories $H1$ and $H2$ such that $H1$ is equivalent to
$H2$ and $H1$ respects conflict order of $H2$, i.e., $\prec_{H1}^{CO} \subseteq 
\prec_{H2}^{CO}$. Then, $\prec_{H1}^{CO} = \prec_{H2}^{CO}$. 
\end{lemma}

\begin{proof}
Here, we have that $\prec_{H1}^{CO} \subseteq \prec_{H2}^{CO}$. In order to 
prove $\prec_{H1}^{CO} = \prec_{H2}^{CO}$,
we have to show that $\prec_{H2}^{CO} \subseteq \prec_{H1}^{CO}$. We prove this 
using contradiction. 
Consider two events $p,q$ belonging to transaction $T1,T2$ respectively in $H2$ 
such that $(p,q)
\in \prec_{H2}^{CO}$ but $(p,q) \notin \prec_{H1}^{CO}$. Since the events of 
$H2$ and $H1$ are same,
these events are also in $H1$. This implies that the events $p, q$ are also 
related by $CO$ in $H1$.
Thus, we have that either $(p,q) \in \prec_{H1}^{CO}$  or $(q,p) \in 
\prec_{H1}^{CO}$. But from our 
assumption, we get that the former is not possible. Hence, we get that 
$(q,p) \in \prec_{H1}^{CO} \Rightarrow (q,p) \in \prec_{H2}^{CO}$. But we 
already have that $(p,q) \in \prec_{H2}^{CO}$.
This is a contradiction. 
\end{proof}

\begin{lemma}
\label{lem:eqv-legal}
Let $H1$ and $H2$ be equivalent histories such that 
$\prec_{H1}^{CO} = \prec_{H2}^{CO}$. 
Then $H1$ is \legal{} iff $H2$ is \legal. 
\end{lemma}

\begin{proof}
It is enough to prove the `if' case, and the `only if' case will follow
from symmetry of the argument. 
Suppose that $H1$ is \legal{}. 
By contradiction, assume that $H2$ is not \legal, i.e.,
there is a read \op{} $r_j(x,v)$ (of transaction $T_j$) in $H2$ with
\lastw{} as $c_k$ (of transaction $T_k$) and $T_k$ writes $u \neq v$
to $x$, i.e $w_k(x, u) \in \evts{T_k}$. 
Let $r_j(x,v)$'s \lastw{} in $H1$ be $c_i$ of $T_i$. 
Since $H1$ is legal, $T_i$ writes $v$ to $x$, i.e $w_i(x, v) \in
\evts{T_i}$. 

Since $\evts{H1} = \evts{H2}$, we get that $c_i$ is also in $H2$, and
$c_k$ is also in $H1$. 
As $\prec_{H1}^{CO} = \prec_{H2}^{CO}$,
we get $c_i <_{H2} r_j(x, v)$ and $c_k <_{H1} r_j(x, v)$. 

Since $c_i$ is the \lastw{} of $r_j(x,v)$ in $H1$ we derive that 
$c_k <_{H1} c_i$ and, thus, $c_k <_{H2} c_i <_{H2} r_j(x, v)$.
But this contradicts the assumption that $c_k$ is the \lastw{} of
$r_j(x,v)$ in $H2$.
Hence, $H2$ is legal. 
\cmnt{
from the w-r conflict order we get that 
\begin{equation}
\label{eq:wr}
c_i <_{H2} r_j(x, v)
\end{equation}

Now, we have two cases based on the ordering of $c_k$ and $r_j$ in $H1$:

\begin{itemize}
\item[] Case (1) $r_j(x, v) <_{H1} c_k$: Since $\prec_{H1}^{CO} = 
\prec_{H2}^{CO}$, from r-w conflict order we get that $r_j <_{H2} c_k$. This 
implies that $c_k$ can not be $r_j$'s \lastw{} in $H2$ which is a 
contradiction. 

\item[] Case (2) $c_k <_{H1} r_j(x, v)$: From w-r conflict order equivalence of 
$\prec_{H1}^{CO} = \prec_{H2}^{CO}$, we get that $c_k <_{H2} r_j$. Here, we 
again have two cases based on the ordering of $c_i$ and $c_k$ in $H1$, Case 
(2.1) $c_i <_{H1} c_k$: In this case, we have that $c_i <_{H1} c_k <_{H1} r_j$ 
which implies that $c_i$ is not $r_j$'s \lastw{} in $H2$, a contradiction. Case 
(2.2) $c_k <_{H1} c_i$: From w-w conflict order equivalence of $\prec_{H1}^{CO} 
= \prec_{H2}^{CO}$, we get that $c_k <_{H2} c_i$. Combining this with 
\eqnref{wr}, we have that $c_k <_{H2} c_i <_{H2} r_j$. This implies that $c_k$ 
can not be $r_j$'s \lastw{} in $H2$ which is a contradiction.
\end{itemize}
Thus in all the cases, we get that $c_k$ can not be $r_j$'s
\lastw. This implies that $H2$ is legal. 
}
\end{proof}
From the above lemma we get the following interesting corollary.

\begin{corollary}
\label{cor:coop-legal}
Every \coop{} history $H$ is \legal{} as well.
\end{corollary}
Based on the conflict graph construction,
we have the following graph characterization for \coop.

\begin{theorem}
\label{thm:graph}
A \legal{} history $H$ is \coop{} iff $\cg{H}$ is acyclic. 
\end{theorem}

\begin{proof}

\noindent
\textit{(Only if)} If $H$ is \coop{} and legal, then
$\cg{H}$ is acyclic:
Since $H$ is \coop{}, there exists a legal t-sequential
history $S$ equivalent to $\overline{H}$ and $S$ respects
$\prec_{H}^{RT}$ and $\prec_{H}^{CO}$. Thus from the conflict graph
construction we have that $\cg{\overline{H}} (= \cg{H})$ is a sub
graph of $\cg{S}$. Since $S$ is sequential, it can be inferred that
$\cg{S}$ is acyclic. Any sub graph of an acyclic graph is also
acyclic. Hence $\cg{H}$ is also acyclic. 

\vspace{1mm}
\noindent
\textit{(if)} If $H$ is \legal{} and $\cg{H}$ is acyclic then $H$ is
\coop: 
Suppose that $\cg{H}=\cg{\overline{H}}$ is acyclic. Thus we can
perform a topological sort on the vertices of the graph and obtain a
sequential order. Using this order, we can obtain a sequential
schedule $S$ that is equivalent to $\overline{H}$.
Moreover, by construction, $S$ respects $\prec_{H}^{RT} =
\prec_{\overline{H}}^{RT}$ 
and $\prec_{H}^{CO} = \prec_{\overline{H}}^{CO}$. 

Since every two events related by the conflict relation (w-w, r-w, or
w-r)in $S$ are also related by $\prec_{\overline{H}}^{CO}$, we obtain 
$\prec_{S}^{CO} = \prec_{\overline{H}}^{CO}$. 
Since $H$ is legal, $\overline{H}$ is also legal. 
Combining this with \lemref{eqv-legal}, we get that $S$ is also
\legal. 
This satisfies all the conditions necessary for $H$ to be \coop. 
\end{proof}

\subsection{Proof of Opacity for CaPR+ Algorithm}
\label{sec:proof}

In this section, we will describe some of the properties of $CaPR^{+}$ algorithm and then prove that it 
satisfies opacity. In our implementation, only the read and \tryc{} \op{s} 
access the memory. Hence, we call these \op{s} as \memop{s}. The main idea 
behind our algorithm is as follows: Consider a live transaction $T_i$ which has 
read a value $u$ for \tobj{} $x$. Suppose a transaction $T_j$ writes a value $v$ 
to \tobj{} $x$ and commits. When $T_i$ executes the next \memop{} (after the 
$c_j$), $T_i$ is rolled back to the step before the read of $x$. We denote that 
$T_j$ has \textit{invalidated} the $T_i$'s read of $x$. Transaction $T_i$ then 
reads $x$ again.

The following example illustrates this idea. Consider the history $H1: r_1(x, 0) 
r_1(y, 0) r_2(x, 0) r_1(z, 0) \\
w_2(y, 5) c_2 w_1(x,5)$. 
In this history, 
when $T_1$ performs any other \memop{} such as a read after $C-$, it will then 
be rolled back to the step $r_1(y)$ causing it to read $y$ again. 

\begin{figure}[h]
\includegraphics[width=8cm]{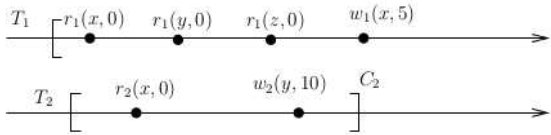}
\caption{Pictorial representation of a History $H1$}
 \label{fig:ex1}
\end{figure}

Thus as explained, in our algorithm, when a transaction's read is invalidated, 
it does not abort but rather gets rolled back. In the worst case, it could get 
rolled back to the first step of the transaction which is equivalent to the 
transaction being aborted and restarted. Thus with this algorithm, a history 
will consist only of incomplete (live) and committed transactions. 

To precisely capture happenings of the algorithm and to make it consistent with 
the model we discussed so far, we modify the representation of the transactions 
that are rolled back. Consider a transaction $T_i$ which has read $x$. Suppose 
another transaction $T_j$ that writes to $x$ and then commits. Thus, when $T_i$ 
performs its next \memop, say $m_i$ (which could either be a read or commit 
\op), it will be rolled back. We capture this rollback \op{} in the history as 
two transactions: $T_{i.1}$ and $T_{i.2}$.

Here, $T_{i.1}$ represents all the successful \op{s} of transaction $T_i$ until 
it executed the \memop{} $m_i$ which caused it to roll back (but not including 
that $m_i$). Transaction $T_{i.1}$ is then terminated by an abort \op{} 
$a_{i.1}$. Then, after transaction $T_j$ has committed transaction $T_{i.2}$ 
begins. Unlike $T_{i.1}$ it is incomplete. It also consists of all same \op{s} 
of $T_{i.1}$ until the read on $x$. $T_{i.2}$ reads the latest value  of the 
\tobj{} $x$ again since it has been invalidated by $T_j$. It then executes 
future steps which could depend on the read of $x$. With this modification, the 
history consists of committed, incomplete as well as aborted transactions (as 
discussed in the model).

In reality, the transaction $T_i$ could be rolled back multiple times, say $n$. 
Then the history $H$ would contain events from transactions $T_{i.1}, T_{i.2}, 
T_{i.3} .... T_{i.n}$. But it must be noted that all the invocations of $T_i$ 
are related by real-time order. Thus, we have that $T_{i.1} \prec^{RT}_{H} 
T_{i.2} \prec^{RT}_{H} T_{i.3} .... \prec^{RT}_{H} T_{i.n}$

With this change in the model, the history $H1$ is represented as follows, $H2: 
r_{1.1}(x, 0) r_{1.1}(y, 0) r_{2.1}(x, 0) \\
r_{1.1}(z, 0) w_{2.1}(y, 5) c_{2.1} w_1(x,5) a_{1.1} r_{1.2}(x, 0) r_{1.2}(y, 
10)$. 

For simplicity, from now on in histories, we will denote a transaction with 
greek letter subscript such as $\alpha, \beta, \gamma$ etc regardless of whether 
it is invoked for the first time or has been rolled back. Thus in our 
representation, transaction $T_{i.1}, T_{i.2}$ could be denoted as $T_\alpha, 
T_\gamma$ respectively.


\begin{figure*}
\includegraphics[width=16cm]{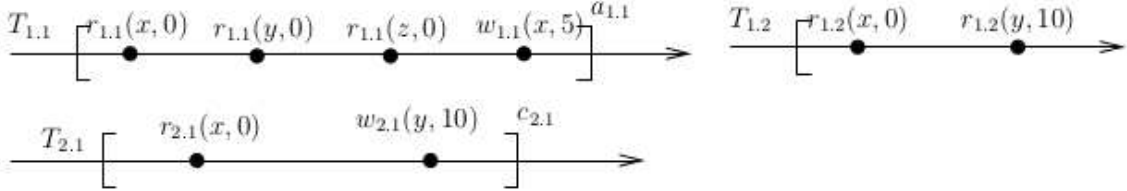}
\caption{Pictorial representation of the modified History $H2$}
 \label{fig:ex2}
\end{figure*}

We will now prove the correctness of this algorithm. We start by describing a 
property that captures the basic idea behind the working of the algorithm. 

\begin{property}
\label{prop:main-algo}
Consider a transaction $T_i$ that reads \tobj{} $x$. Suppose another transaction 
$T_j$ writes to $x$ and then commits. In this case, the next \memop{} (read or 
\tryc) executed by $T_i$ after $c_j$ returns abort (since the read of $x$ by 
$T_i$ has been invalidated). 
\end{property}

For a transaction $T_i$, we define the notion of \textit{successful final \memop 
(\sfm)}. As the name suggests, it is the last successfully executed \memop{} of 
$T_i$. If $T_i$ is committed, then $\sfm_i = c_i$. If $T_i$ is aborted, then 
$\sfm_i$ is the last \memop, in this case a read \op,  that returned $ok$ before 
being aborted. 

For proving the correctness, we use the graph characterization of \coopty{} 
described in \secref{co}. 

Consider a history $H_{capr}$ generated by the CaPR 
algorithm. Let $CG(H_{capr})$ be the conflict graph of $H_{capr}$. We show that this graph denoted, $g_{capr}$, is acyclic.

\begin{lemma}
\label{lem:path-sfm}
Consider a path $p$ in $g_{capr}$ abstracted as: $T_{\alpha 1} 
\rightarrow T_{\alpha 2} \rightarrow .... \rightarrow T_{\alpha k}$. Then, 
$\sfm_{\alpha 1} <_{H_{capr}} \sfm_{\alpha 2} <_{H_{capr}} .... <_{H_{capr}} 
\sfm_{\alpha k}$. 
\end{lemma}

\begin{proof}
We prove this using induction on k. 

\vspace{1mm}
\noindent
\textit{Base Case, $k=2$.} In this case the path consists of only one edge 
between transactions $T_{\alpha 1}$ and $T_{\alpha 2}$. Let us analyse the 
various types of edges possible:

\begin{itemize}
\item \textit{real-time edge:} This edge represents real-time. In this case 
$T_{\alpha 1} \prec_{H_{capr}}^{RT} T_{\alpha 2}$. Hence, we have that 
$\sfm_{\alpha 1} <_{H_{capr}} \sfm_{\alpha 2}$.

\item \textit{w-w edge:} This edge represents w-w order conflict. In this case 
both transactions $T_{\alpha 1}$ and $T_{\alpha 2}$ are committed and 
$\sfm_{\alpha 1} = c_{\alpha 1}$ and $\sfm_{\alpha 2} = c_{\alpha 2}$. Thus, 
from the definition of this conflict, we get that $\sfm_{\alpha 1} <_{H_{capr}} 
\sfm_{\alpha 2}$.

\item \textit{w-r edge:} This edge represents w-r order conflict. In this case, 
$c_{\alpha 1} <_{H_{capr}} r_{\alpha 2}(x, v)$ ($v \neq A$). For transaction 
$T_{\alpha 1}$, $\sfm_{\alpha 1} = c_{\alpha 1}$. For transaction $T_{\alpha 
2}$, either $r_{\alpha 2} <_{H_{capr}} \sfm_{\alpha 2}$ or $r_{\alpha 2} = 
\sfm_{\alpha 2}$. Thus in either case, we get that $\sfm_{\alpha 1} <_{H_{capr}} 
\sfm_{\alpha 2}$. 

\item \textit{r-w edge:} This edge represents r-w order conflict. In this case, 
$r_{\alpha 1}(x, v) <_{H_{capr}} c_{\alpha 2} $ (where $v \neq A$). Thus 
$\sfm_{\alpha 2} = c_{\alpha 2}$. Here, we again have two cases: (a) $T_{\alpha 
1}$ terminates before $T_{\alpha 2}$. In this case, it is clear that 
$\sfm_{\alpha 1} <_{H_{capr}} \sfm_{\alpha 2}$. (b) $T_{\alpha 1}$ terminates 
after $T_{\alpha 2}$ commits. The working of the algorithm is such that, as 
observed in \propref{main-algo}, the next \memop{} executed by $T_{\alpha 1}$ 
after the commit \op{} $c_{\alpha 2}$ returns abort. From this, we get that the 
last successful \memop{} executed by $T_{\alpha 1}$ must have executed before 
$c_{\alpha 2}$. Hence, we get that $\sfm_{\alpha 1} <_{H_{capr}} \sfm_{\alpha 
2}$.
\end{itemize}
Thus in all the cases, the base case holds.

\vspace{1mm}
\noindent
\textit{Induction Case, $k=n > 2$.} In this case the path consists of series of 
edges starting from transactions $T_{\alpha 1}$ and ending at $T_{\alpha n}$. 
From our induction hypothesis, we know that it is true for $k = n-1$. Thus, we 
have that $\sfm_{\alpha 1} <_{H_{capr}} \sfm_{\alpha (n-1)}$. Now consider the 
transactions $T_{\alpha (n-1)}, T_{\alpha n}$ which have an edge between them. 
Using the arguments similar to the base case,  we can prove that $\sfm_{\alpha 
(n-1)} <_{H_{capr}} \sfm_{\alpha n}$. Thus, we have that $\sfm_{\alpha 1} 
<_{H_{capr}} \sfm_{\alpha n}$.

In all the cases, we have that $\sfm_{\alpha 1} <_{H_{capr}} \sfm_{\alpha n}$. 
Hence, proved. 
\end{proof}

Using \lemref{path-sfm}, we show that $g_{capr}$ is acyclic.

\begin{lemma}
\label{lem:acyclic}
Graph, $g_{capr}$ is acyclic.
\end{lemma}

\begin{proof}
We prove this by contradiction. Suppose that $g_{capr}$ is cyclic. Then there is 
a cycle going from $T_{\alpha 1} \rightarrow T_{\alpha 2} \rightarrow .... 
\rightarrow T_{\alpha k} \rightarrow T_{\alpha 1}$. 

From \lemref{path-sfm}, we get that $\sfm_{\alpha 1} \rightarrow \sfm_{\alpha 2} 
\rightarrow .... \rightarrow \sfm_{\alpha k} \rightarrow \sfm_{\alpha 1}$ which 
implies that $\sfm_{\alpha 1} \rightarrow \sfm_{\alpha 1}$. Hence, the 
contradiction. 
\end{proof}

\begin{theorem}
\label{thm:graph}
All histories generated by $CaPR^+$ are co-opaque and hence, $Capr^+$ satisfies the property of opacity. 
\end{theorem}
\begin{proof}
Proof follows from Theorem 5 and Lemma 8.
\end{proof}


\section{Conclusion}
In this paper, we have described $CaPR^{+}$, an enhanced $CaPR$ algorithm and proved its opacity. We have also implemented the same and tested its 
performance. While it shows good performance for transactions that take time, its performance for small transactions
shows overhead which is obvious.  
A thorough comparison with STAMP benchmarks with varying transactions has been done and shows good results. 
This will be reported elsewhere. Further, we have been working on several optimizations
like integrating both partial rollback and abort mechanisms in the same implementation to exploit the benefits of both mechanisms, 
and also integrate with it the contention management. 
\end{document}

%% file: macros.tex
\newcommand{\punt}[1]{}
\newcommand{\cmnt}[1]{}

\newtheorem{theorem}{Theorem}

\newtheorem{lemma}[theorem]{Lemma}
\newtheorem{corollary}[theorem]{Corollary}

\newtheorem{property}[theorem]{Property}
\newtheorem{definition}[theorem]{Definition}

\newcommand{\secref}[1]{Section~\ref{sec:#1}}

\newcommand{\lemref}[1]{Lemma~\ref{lem:#1}}

\newcommand{\eqnref}[1]{Eqn(\ref{eq:#1})}

\newcommand{\propref}[1]{Property~\ref{prop:#1}}

\newcommand{\id}[1]{\mbox{\textit{#1}}}

\newcommand{\op} {operation}
\newcommand{\termop} {terminal operation}
\newcommand{\memop} {memory operation}
\newcommand{\tobj} {t-object}
\newcommand{\evts}[1] {evts(#1)}

\newcommand{\comm}{\textit{committed}}
\newcommand{\aborted}{\textit{aborted}}
\newcommand{\txns}{\textit{txns}}

\newcommand{\begtrans} {begin\_tran}
\newcommand{\tryc} {tryC}
\newcommand{\trya}{\textit{tryA}}

\newcommand{\Wset}{\textit{Wset}}

\newcommand{\valid} {valid}

\newcommand{\legal} {legal}

\newcommand{\lastw} {lastWrite}
\newcommand{\lwrite}[2] {#2.lastWrite(#1)}

\newcommand{\co} {CO}

\newcommand{\opty} {opacity}
\newcommand{\coop} {co-opaque}
\newcommand{\coopty} {co-opacity}
\newcommand{\cg}[1] {CG(#1)}

\newcommand{\sfm} {sfm}

